\documentclass[11pt]{article}
\usepackage[margin=1in]{geometry}
\usepackage{amsmath,amssymb,amsthm,mathtools}
\usepackage{graphicx}
\usepackage{xcolor}
\usepackage{enumitem}
\usepackage{booktabs}
\usepackage{microtype}
\usepackage{hyperref}
\hypersetup{colorlinks=true, linkcolor=blue, citecolor=blue, urlcolor=blue}
\usepackage{thmtools}

\newtheorem{definition}{Definition}
\newtheorem{assumption}{Assumption}
\newtheorem{theorem}{Theorem}
\newtheorem{lemma}{Lemma}
\newtheorem{proposition}{Proposition}

\title{Exam Readiness Index (ERI): A Theoretical Framework for a Composite, Explainable Index}
\author{Ananda Prakash Verma\\ \texttt{anandaprakashverma@gmail.com}\\}
\date{\today}

\begin{document}
\maketitle

\begin{abstract}
We present a theoretical framework for an \emph{Exam Readiness Index} (ERI): a composite, blueprint-aware score $R \in [0,100]$ that summarizes a learner's readiness for a high-stakes exam while remaining interpretable and actionable. The ERI aggregates six signals---Mastery $(M)$, Coverage $(C)$, Retention $(\mathcal{R})$, Pace $(P)$, Volatility $(V)$, and Endurance $(E)$---each derived from a stream of practice and mock-test interactions. We formalize axioms for component maps and the composite, prove monotonicity, Lipschitz stability, and bounded drift under blueprint re-weighting, and show existence and uniqueness of the optimal linear composite under convex design constraints. We further characterize confidence bands via blueprint-weighted concentration and prove compatibility with prerequisite-admissible curricula (Knowledge Space / learning-space constraints). The paper focuses on theory; empirical study is left to future work.
\end{abstract}

\section{Introduction}
High-stakes exam preparation hinges on answering four questions: \emph{Am I ready now? Ready for what? What's limiting me? What should I do next?} We propose a principled, explainable readiness index that is \emph{blueprint-aware}: it respects an official syllabus with section/topic weights and pacing rules, and it decomposes into human-interpretable components that align with cognitive and operational constructs (ability, coverage, forgetting, timing, stability, endurance).

This paper makes four contributions:
\begin{enumerate}[leftmargin=*]
    \item \textbf{Axiomatic components.} We specify component maps from interaction streams to $[0,1]$ with regularity properties (boundedness, Lipschitz continuity, monotonicity).
    \item \textbf{Blueprint-aware composite.} We define a linear ERI $R = \sum_{i} \alpha_i X_i$ with design constraints and show stability and interpretability guarantees.
    \item \textbf{Confidence and drift.} We derive blueprint-weighted concentration bounds for component estimates and a bound on ERI drift under blueprint version change.
    \item \textbf{Prerequisite compatibility.} We provide a sufficient condition ensuring ERI-driven recommendations respect admissibility in a learning space (outer-fringe).
\end{enumerate}

\paragraph{Scope.} Our goal is a deployable theoretical scaffold. We do not perform parameter estimation or empirical validation; those are orthogonal and deferred to future work.

\section{Related Work}
\paragraph{Knowledge Space Theory (KST).} KST models feasible \emph{knowledge states} as union-closed families of item sets with natural learning properties (accessibility, well-gradedness)~\cite{doignon,learningSpaces}. ERI is complementary: it provides a continuous readiness summary while remaining compatible with KST's admissible-next (outer fringe) principle (Section~\ref{sec:admissibility}).

\paragraph{Item Response Theory (IRT).} IRT~\cite{rasch,lord} models a latent ability parameter and item characteristics to predict responses. ERI does not require a single latent trait; instead, it aggregates multiple blueprint-weighted signals. When available, IRT-calibrated difficulties can instantiate the Mastery component $M$.

\paragraph{Spaced Repetition and Retention.} A large body of work—tracing to Ebbinghaus' forgetting curve—shows that spaced retrieval improves long-term retention~\cite{ebbinghaus,cepeda,pavlik}. ERI's Retention component $\mathcal{R}$ formalizes a blueprint-weighted exponential decay that is compatible with trainable schedulers (e.g., Leitner/FSRS-style models).

\section{Blueprints, Signals, and Notation}
Let $\mathcal{S}$ denote sections and $\mathcal{T}$ topics. A \emph{blueprint} $B$ provides nonnegative weights $\{w_t\}_{t\in \mathcal{T}}$ with $\sum_{t} w_t = 1$, a duration $T_{\mathrm{exam}}$, marking rules, and section-level pace targets $\{\tau_s\}_{s\in\mathcal{S}}$ (seconds per item). Interaction logs induce per-topic sequences $\mathcal{D}_t$ (attempt outcomes, timestamps, times), and per-section sequences $\mathcal{D}_s$.

We define six normalized component maps (all into $[0,1]$):
\begin{align*}
M &:= \sum_{t\in\mathcal{T}} w_t \, m_t(\mathcal{D}_t), \quad
C := \sum_{t\in\mathcal{T}} w_t \, c_t(\mathcal{D}_t), \quad
\mathcal{R} := \sum_{t\in\mathcal{T}} w_t \, r_t(\mathcal{D}_t),\\
P &:= \frac{1}{|\mathcal{S}|}\sum_{s\in\mathcal{S}} p_s(\mathcal{D}_s), \quad
V := v(\mathcal{D}), \quad
E := e(\mathcal{D}).
\end{align*}

\begin{assumption}[Component axioms]\label{ass:components}
For each topic $t$ and section $s$, the maps satisfy:
\begin{enumerate}[label=(A\arabic*),leftmargin=*]
    \item \textbf{Normalization:} $m_t,c_t,r_t,p_s,v,e \in [0,1]$ for any admissible data stream.
    \item \textbf{Monotonicity:} $m_t$ is nondecreasing in difficulty-adjusted success; $c_t$ is nondecreasing in (evidence-weighted) coverage; $r_t$ is nonincreasing in recency gaps; $p_s$ is nonincreasing in pace deviation; $v$ is nonincreasing in score variance; $e$ is nonincreasing in late-session degradation.
    \item \textbf{Lipschitz regularity:} For appropriate metrics $d_t,d_s$ on streams, there exist $L_m,L_c,L_r,L_p,L_v,L_e$ with
    \[
        |m_t(\mathcal{D}_t)-m_t(\mathcal{D}'_t)| \le L_m\, d_t(\mathcal{D}_t,\mathcal{D}'_t), \quad \text{etc.}
    \]
    \item \textbf{Blueprint separability:} Topic-level components aggregate linearly with weights $\{w_t\}$.
\end{enumerate}
\end{assumption}

\paragraph{Instantiation examples (non-essential).} One may take $r_t(\mathcal{D}_t) = \exp(-\lambda_t \Delta_t)$ with $\Delta_t$ the time since last successful retrieval and $\lambda_t>0$; map section pace to $p_s = \phi(z_s)$ for a clipped linear $\phi$ of a pace $z$-score; and define coverage gates by evidence indicators. Our results require only Assumption~\ref{ass:components}.

\section{The ERI Composite and Design Constraints}
We define the readiness index as a convex combination of components:
\begin{equation}\label{eq:eri}
R(\mathcal{D};B,\alpha) \;=\; \alpha_M M + \alpha_C C + \alpha_R \mathcal{R} + \alpha_P P + \alpha_V V + \alpha_E E,
\end{equation}
with weights $\alpha \in \mathbb{R}^6_{\ge 0}$ and $\sum_i \alpha_i = 1$.

\subsection{Axioms for the composite}
\begin{definition}[ERI axioms]\label{def:axioms}
A composite $R$ is \emph{valid} if it satisfies: (i) \textbf{Boundedness} $R\in[0,1]$; (ii) \textbf{Monotonicity} in each argument; (iii) \textbf{Blueprint coherence} $\partial R/\partial w_t \ge 0$ whenever $m_t,c_t,r_t$ weakly increase; (iv) \textbf{Scale-invariance} to monotone reparameterizations of individual components (achieved by pre-normalization).
\end{definition}

\begin{proposition}[Validity of the linear ERI]\label{prop:valid}
Under Assumption~\ref{ass:components}, the linear $R$ in \eqref{eq:eri} is valid per Definition~\ref{def:axioms}.
\end{proposition}

\begin{proof}
Boundedness and monotonicity follow from nonnegativity and normalization. Blueprint coherence follows by linearity and separability. Scale-invariance holds as each component is normalized to $[0,1]$.
\end{proof}

\subsection{Weight selection as convex design}
Let $\mathcal{C}$ denote design constraints (e.g., minimum emphasis on mastery and coverage, fairness caps per section), and $J(\alpha)$ a strictly convex penalty (e.g., quadratic deviation from a prior $\alpha^0$ and entropy regularization for spread). Consider
\begin{equation}\label{eq:design}
\min_{\alpha \in \Delta_5} \; J(\alpha) \quad \text{s.t. } \alpha \in \mathcal{C},
\end{equation}
where $\Delta_5$ is the 5-simplex.

\begin{theorem}[Existence and uniqueness]\label{thm:unique}
If $J$ is strictly convex and $\mathcal{C}$ is convex and nonempty, the solution $\alpha^\star$ to \eqref{eq:design} exists and is unique.
\end{theorem}
\begin{proof}
Standard convex-optimization result.
\end{proof}

\begin{lemma}[Lagrangian KKT form]\label{lem:kkt}
Let $J(\alpha)=\tfrac12\|\alpha-\alpha^0\|_2^2 - \eta \sum_i \alpha_i \log \alpha_i$. Then the KKT conditions yield
\[
\alpha_i^\star \propto \exp\left( \alpha_i^0/\eta - \lambda_i \right),
\]
with multipliers $\lambda_i$ enforcing the linear constraints in $\mathcal{C}$ and the simplex.
\end{lemma}

\section{Stability and Confidence}
\subsection{Lipschitz stability}
\begin{theorem}[ERI Lipschitz constant]\label{thm:lipschitz}
Let $d(\mathcal{D},\mathcal{D}')$ be a metric that dominates per-topic/section metrics. Under Assumption~\ref{ass:components},
\[
|R(\mathcal{D}) - R(\mathcal{D}')| \le \Big(\alpha_M L_M + \alpha_C L_C + \alpha_R L_R + \alpha_P L_P + \alpha_V L_V + \alpha_E L_E\Big) \, d(\mathcal{D},\mathcal{D}').
\]
\end{theorem}
\begin{proof}
Triangle inequality and linearity.
\end{proof}

\subsection{Blueprint drift}
Suppose a blueprint update changes weights from $w$ to $w'$ with total-variation $\delta=\tfrac12\sum_t |w_t-w'_t|$.

\begin{proposition}[Bounded ERI drift under re-weighting]\label{prop:drift}
With components fixed,
\[
|R(\mathcal{D};w)-R(\mathcal{D};w')| \;\le\; \delta \cdot \max\{ \|m\|_\infty+\|c\|_\infty+\|r\|_\infty,\, 1 \}.
\]
If all topic components lie in $[0,1]$, the drift is at most $3\delta$ for the topic-aggregated terms.
\end{proposition}
\begin{proof}
By linearity, the difference is a weighted sum of bounded topic terms, bounded by the $\ell_1$-distance of weights.
\end{proof}

\subsection{Confidence band via blueprint-weighted concentration}
Let $\widehat{m}_t$ estimate $m_t$ from $n_t$ independent item outcomes in topic $t$ with bounded range. By Hoeffding's inequality,
\[
\Pr\Big(|\widehat{m}_t-m_t| \ge \epsilon_t\Big) \le 2\exp(-2 n_t \epsilon_t^2).
\]
Aggregating with blueprint weights gives, for any $\epsilon>0$,
\begin{equation}\label{eq:eri-conc}
\Pr\Big(|\widehat{M}-M| \ge \epsilon\Big) \le 2\exp\!\left(-\frac{2\epsilon^2}{\sum_t w_t^2/n_t}\right).
\end{equation}
Analogous bounds hold for $\widehat{C},\widehat{\mathcal{R}}$ when estimated from bounded per-topic statistics. A union bound yields an overall ERI confidence band whose half-width shrinks with blueprint-weighted effective sample size $\left(\sum_t w_t^2/n_t\right)^{-1}$.

\section{Prerequisite-Admissible Recommendations}\label{sec:admissibility}
Let $(Q,\mathcal{K})$ be a learning space (union-closed, accessible) over items/skills with \emph{outer fringe} $F^+(X)=\{q\notin X: X\cup\{q\}\in\mathcal{K}\}$. A recommendation policy is \emph{admissible} if it selects only items whose concepts lie in $F^+(X)$ or their immediate prerequisites.

\begin{assumption}[ERI-compatible gating]\label{ass:gating}
Action maps use only concepts with nonzero coverage evidence and admissible with respect to the current state estimate $X$.
\end{assumption}

\begin{proposition}[ERI actions respect admissibility]\label{prop:admissible}
Under Assumption~\ref{ass:gating}, any action chosen by minimizing a separable surrogate loss of the form $\sum_{t} \psi_t(1-M_t,C_t,\mathcal{R}_t)$ over admissible concepts respects the outer fringe and cannot introduce prerequisite violations.
\end{proposition}
\begin{proof}
By construction, the feasible set is $F^+(X)$; separability preserves feasibility.
\end{proof}

\section{EDGE Integration (Future Work)}
The ERI is readily composable with a broader adaptive framework such as \emph{EDGE} (Evaluate $\rightarrow$ Diagnose $\rightarrow$ Generate $\rightarrow$ Exercise) ~\cite{verma2025edge}. Concretely:
\begin{itemize}[leftmargin=*]
    \item \textbf{Inputs to EDGE.} The vector $(M,C,\mathcal{R},P,V,E)$ and blueprint weights act as state features for schedulers and item selection. Low-$\mathcal{R}$ concepts are prioritized for spaced retrieval; low-$M$/$C$ concepts trigger targeted drills.
    \item \textbf{Outputs back to ERI.} EDGE's counterfactual drills and retrieval sessions update topic-level statistics, which in turn update $M$, $C$, and $\mathcal{R}$; pace-aware exercises update $P$; session composition affects $E$ and $V$ stability.
    \item \textbf{Governance.} ERI provides an interpretable, bounded readiness signal and a confidence band; EDGE supplies the \emph{actions} that move the signal. The two systems remain compatible with learning-space admissibility by gating recommendations to the outer fringe.
\end{itemize}
A separate manuscript will detail EDGE's learning-theoretic underpinnings; here we note only the interfaces and invariants required for safe composition.

\section{Limitations and Future Work}
Our framework assumes component regularity and separability; while natural, they may be restrictive for highly entangled domains. Confidence bands rely on independence and boundedness assumptions; relaxing them to account for item overlap and adaptive sampling is future work. Empirical validation, parameter learning, and outcome calibration are deliberately out-of-scope.

\section{Conclusion}
We presented a blueprint-aware, axiomatic readiness index with theoretical guarantees: monotonicity, Lipschitz stability, bounded drift under re-weighting, unique weight design under convex constraints, and compatibility with prerequisite-admissible recommendations. The framework is deployment-ready and can later accommodate richer component estimators, empirical calibration, and integration with EDGE-style adaptive schedulers.

\appendix
\section{Appendix A: Additional Lemmas and Proofs}
\begin{lemma}[Component aggregation Lipschitzness]
If each topic map $m_t$ is $L_m$-Lipschitz w.r.t. $d_t$ and $\sum_t w_t=1$, then $M$ is $L_m$-Lipschitz w.r.t. $d=\sum_t w_t d_t$.
\end{lemma}
\begin{proof}
Immediate by convexity and triangle inequality.
\end{proof}

\begin{lemma}[Explicit Lipschitz constants for exponential retention]
If $r_t(\mathcal{D}_t)=\exp(-\lambda_t \Delta_t)$ with $\Delta_t$ the time since last success and $|\Delta_t-\Delta_t'|\le \delta_t$, then $|r_t-r_t'|\le \lambda_t \delta_t$ for $\delta_t$ small; hence $L_r\le \max_t \lambda_t$.
\end{lemma}

\begin{proposition}[Identifiability under ordinal constraints]
Suppose we observe only pairwise orderings of readiness for a family of profiles. If the set of profiles has full-dimensional convex hull in $\mathbb{R}^6$ and we require $R$ to be linear and monotone, then $\alpha$ is identifiable up to a positive scalar (fixed by the simplex constraint).
\end{proposition}
\begin{proof}
Follows from uniqueness of separating hyperplanes consistent with strict orderings on a full-dimensional set.
\end{proof}

\section{Appendix B: Lagrangian Details}
Consider $\min_{\alpha\in\Delta_5} \tfrac12\|\alpha-\alpha^0\|_2^2 - \eta \sum_i \alpha_i \log \alpha_i$ with linear constraints $A\alpha \le b$. The Lagrangian is
\[
\mathcal{L}(\alpha,\lambda,\nu) = \tfrac12\|\alpha-\alpha^0\|_2^2 - \eta \sum_i \alpha_i \log \alpha_i + \lambda^\top (A\alpha-b) + \nu ( \mathbf{1}^\top \alpha -1).
\]
Stationarity gives $(\alpha_i-\alpha_i^0) - \eta (1+\log \alpha_i) + (A^\top \lambda)_i + \nu = 0$, yielding the exponential-form solution stated in Lemma~\ref{lem:kkt}.

\section{Appendix C: Knowledge Spaces and Admissibility}
Let $(Q,\mathcal{K})$ be a \emph{learning space} (union-closed, accessible, well-graded). If the action generator restricts to $F^+(X)$ and surrogates are separable, ERI-driven policies are admissible. Extensions include \emph{quasi-ordinal} relaxations for noisy prerequisites.

\section*{Acknowledgements}
We thank colleagues for helpful discussions. This paper intentionally avoids empirical validation; any errors are our own.

\end{document}